\documentclass[11pt,english]{article}
\usepackage{color}
\usepackage[T1]{fontenc}
\usepackage[latin9]{inputenc}
\usepackage[a4paper]{geometry}
\usepackage{textcomp}
\usepackage{amsthm}
\usepackage{amsmath}
\usepackage{amssymb}
\usepackage{esint}
\usepackage{bbm}
\usepackage{hyperref}
\usepackage{babel}

\usepackage{amscd}

\usepackage{amsfonts}
\usepackage{amsthm}
\usepackage{mathrsfs}
\usepackage{dsfont}

\usepackage{mathtools}

\usepackage{enumerate}
\usepackage{bm}
\usepackage{bbm}
\usepackage{verbatim}

\usepackage{enumitem}
\newlist{abbrv}{itemize}{1}
\setlist[abbrv,1]{label=,labelwidth=1.2in,align=parleft,itemsep=0.1\baselineskip,leftmargin=!}

\makeatletter

\DeclareFontEncoding{LGR}{}{}
\DeclareTextSymbol{\~}{LGR}{126}

\newcommand{\tr}{\textnormal{Tr}}
\newcommand{\vA}{\boldsymbol{\hat{A}}_{\kappa}}

\newcommand{\vAc}{\boldsymbol{A}_{\kappa}}

\newcommand{\vE}{\boldsymbol{\hat{E}}_{\kappa}}

\newcommand{\vEc}{\boldsymbol{E}_{\kappa}}

\newcommand{\vF}{\boldsymbol{\hat{F}}}
\newcommand{\vFp}{\boldsymbol{\hat{F}}^{+}}

\newcommand{\vFm}{\boldsymbol{\hat{F}}^{-}}
\newcommand{\vFc}{\boldsymbol{F}}
\newcommand{\vFpc}{\boldsymbol{F}^{+}}
\newcommand{\vFmc}{\boldsymbol{F}^{-}}

\newcommand{\vep}{\boldsymbol{\epsilon}}

\newcommand{\abs}[1]{| #1 |}

\newcommand{\vj}{\boldsymbol{j}}

\newcommand{\ND}{\mathcal{N}_{\delta}}

\newcommand{\scp}[2]{\big\langle #1 , #2 \big\rangle}

\newcommand{\SCP}[2]{\big\langle #1 , #2 \big\rangle}

\newcommand{\bra}[1]{\langle #1 |}

\newcommand{\ket}[1]{| #1 \rangle}

\newcommand{\norm}[1]{\left\| #1 \right\|}

\renewcommand{\Re}{\mathrm{Re}}
\renewcommand{\Im}{\mathrm{Im}}

\newcommand{\id}{\mathbbm{1}}

\newcommand{\be}{\begin{equation}}
\newcommand{\ee}{\end{equation}}

\newtheorem{theorem}{Theorem}[section]
\newtheorem{lemma}[theorem]{Lemma}

\newtheorem{corollary}[theorem]  {Corollary}
\newtheorem{remark}[theorem]  {Remark}
\newtheorem{definition}[theorem] {Definition}
\newtheorem{proposition}[theorem]{Proposition}
\newtheorem{assumption}[theorem]{Assumption}

\numberwithin{equation}{section}

\usepackage[colorinlistoftodos]{todonotes}

\usepackage{ulem}

\allowdisplaybreaks[1]

\begin{document}

\title{Derivation of the Maxwell-Schr\"odinger Equations: A note on the infrared sector of the radiation field}

\author{ Marco Falconi\footnote{Dipartimento di Matematica, Politecnico di Milano, Piazza Leonardo da Vinci, 32, 20133, Milano, Italy, E-mail address: {\tt  marco.falconi@polimi.it}} \ and
Nikolai Leopold\footnote{University of Basel, Department of Mathematics and Computer Science, Spiegelgasse 1, 4051 Basel, Switzerland, E-mail address: {\tt nikolai.leopold@unibas.ch}} }

\maketitle

\begin{abstract}
\noindent
We slightly extend prior results about the derivation of the Maxwell-Schr\"odinger equations from the bosonic Pauli-Fierz Hamiltonian. More concretely, we show  that the findings from \cite{LP2020} about the coherence of the quantized electromagnetic field also hold for soft photons with small energies. This is achieved with the help of an estimate from \cite{AFH2022} which proves that the domain of the number of photon operator is invariant during the time evolution generated by the Pauli-Fierz Hamiltonian.
\end{abstract}

\noindent
\textbf{MSC class:} 35Q40, 81Q05, 81V10, 82C10   

\section{Introduction}

In this short paper we derive the Maxwell-Schr\"odinger system of equations as
an effective model describing a Bose-Einstein condensate of charged particles
immersed in a coherent electromagnetic field. More precisely, we prove quantitatively
that the Maxwell-Schr\"odinger system approximates well the many-body quantum evolution generated by the Pauli-Fierz Hamiltonian;
provided that the total number of particles $N$ is large, the particles are initially in a Bose-Einstein condensate and that the
quantum nature of the field -- quantified by the semiclassical parameter $\hslash$ --
is negligible. In particular, we focus on the combined regime $N \sim
\frac{1}{\hslash} \to +\infty$. Equivalently, the same effective dynamics approximates
well $N\to \infty$ bosons weakly interacting with a quantized electromagnetic field,
see the discussion below.

This problem has already been studied by one of the authors, together with
P.\ Pickl, in \cite{LP2020}. 
The main focus here is to build on the
results and techniques introduced there, and to strengthen them by studying
convergence for the photons' reduced density matrix. 
In the previous work the quantum fluctuations around the coherent state of photons have been classified only by means of their energy. 
The extension to the reduced density matrix is physically relevant and mathematically nontrivial because the coherence of photons with small frequencies can not be shown by the energy of the electromagnetic field, due to its massless nature. We often refer to \cite{LP2020} throughout the paper, hopefully striking a good balance between being concise and being
self-contained.

\subsection{The Maxwell-Schr\"odinger System of Equations}
\label{sec:maxw-schr-syst}

The Maxwell-Schr\"odinger system of equations describes the wave function $\varphi$ of
a quantum particle (with a nontrivial charge distribution $\kappa$) interacting
with the classical electromagnetic field, described by the vector potential
$\mathbf{A}$ and the electric field $\mathbf{E} = - \dot{\mathbf{A}}$. We choose the Coulomb gauge
\begin{equation}
  \nabla\cdot \mathbf{A}=0 
\end{equation}
and this makes indeed $\mathbf{A}$ and $\mathbf{E}$ the only dynamical
degrees of freedom of the field. Let us also preliminarily define the current
\begin{equation}
  \label{eq: current}
  \mathbf{j}= 2\bigl(\Im (\bar{\varphi}\nabla\varphi) - \bar{\varphi}(\kappa*\mathbf{A})\varphi\bigr) \; .
\end{equation}
The Maxwell-Schr\"odinger system thus takes the form
\begin{align}
\label{eq:Hartree-Maxwell 2}
\begin{cases}
  &i \partial_t \varphi = \bigl(-i\nabla - (\kappa*\mathbf{A})\bigr)^2\varphi + \mathcal{V}[\varphi] \\[3mm]
  &\partial_t \mathbf{A} = - \mathbf{E}\\[3mm]
  &\partial_t \mathbf{E} = - \Delta \mathbf{A} -\bigl(1-\nabla (\nabla\cdot (\Delta^{-1})\bigr) (\kappa* \mathbf{j})
\end{cases}\quad ,
\end{align}
where $\mathcal{V}[\varphi]$ is an interaction term for the quantum particle. The
choice of gauge, Coulomb's in this case, can be seen as a constraint, for it
is preserved by the Maxwell-Schr\"odinger flow. A typical example for the
particle interaction $\mathcal{V}[\varphi]$ could be
\begin{equation}
  \label{eq: effective particle potential}
  \mathcal{V}[\varphi] = \bigl( W + v* \lvert \varphi  \rvert_{}^{2}\bigr)\varphi\; ,
\end{equation}
where $W$ is an external potential and $(v* \lvert \varphi \rvert_{}^2)\varphi$ a nonlinear term,
usually originating from a microscopic pair interaction. The Cauchy problem
associated to \eqref{eq:Hartree-Maxwell 2} is obtained by fixing an initial datum $(\varphi_0,\mathbf{A}_0,\mathbf{E}_0)$, subjected to the
constraint $\nabla\cdot \mathbf{A}_0=0$. In order to do so, it is convenient to
introduce the complex scalar fields $\bigl(\alpha_0(\cdot ,\lambda)\bigr)_{\lambda=1,2}$ by
defining
\begin{align}
\label{eq: complex-real fields}
&\mathbf{A}_0(x) =  \tfrac{1}{(2\pi)^{3/2}}   \sum_{\lambda=1,2} \int \mathrm{d}^3k \, \frac{1}{\sqrt{2 \abs{k}}} \vep_{\lambda}(k) 
\left( e^{ikx} \alpha_0(k,\lambda) + e^{-ikx} \overline{\alpha_0(k,\lambda)}  \right)\;, \\
& \mathbf{E}_0(x)=  \tfrac{i}{(2\pi)^{3/2}}   \sum_{\lambda=1,2} \int \mathrm{d}^3k  \,  \sqrt{\frac{\abs{k}}{2}} \vep_{\lambda}(k) 
\left( e^{ikx} \alpha_0(k,\lambda) - e^{-ikx}  \overline{\alpha_0(k,\lambda)}  \right)  \; ,
\end{align}
where $\bigl(\vep_{\lambda}(k)\bigr)_{\lambda=1,2}$ are the polarization vectors
satisfying
\begin{equation}
  \vep_{\lambda}(k)\cdot \vep_{\mu}(k) = \delta_{\lambda\mu}\;,\; k\cdot \vep_{\lambda}(k)=0\;,
\end{equation}
that implement the Coulomb gauge. In fact, there is a unique such
decomposition for any time, \textit{i.e.},
\begin{align}
\label{eq: complex-real fields time 1}
&\mathbf{A}(x,t) =  \tfrac{1}{(2\pi)^{3/2}}   \sum_{\lambda=1,2} \int \mathrm{d}^3k \, \frac{1}{\sqrt{2 \abs{k}}} \vep_{\lambda}(k) 
\left( e^{ikx} \alpha_t(k,\lambda) + e^{-ikx} \overline{\alpha_t(k,\lambda)}  \right)\;, \\
\label{eq: complex-real fields time 2}
& \mathbf{E}(x,t)=  \tfrac{i}{(2\pi)^{3/2}}   \sum_{\lambda=1,2} \int \mathrm{d}^3k  \,  \sqrt{\frac{\abs{k}}{2}} \vep_{\lambda}(k) 
\left( e^{ikx} \alpha_t(k,\lambda) - e^{-ikx}  \overline{\alpha_t(k,\lambda)}  \right)  \; ,
\end{align}
that respects both the Coulomb gauge and $\dot{\mathbf{A}}= - \mathbf{E}$. This makes it possible to consider the equivalent system
\begin{align}
\label{eq:Hatree-Maxwell}
\begin{cases}
i \partial_t \varphi_t &=  \left(-i \nabla - \kappa * \boldsymbol{A}(\cdot,t) \right)^2 \varphi_t + \mathcal{V}[\varphi] ,    \\
i \partial_t \alpha_t(k,\lambda)
&= \abs{k} \alpha_t(k,\lambda) - \sqrt{\frac{4 \pi^3}{\abs{k}}}  \mathcal{F}[\kappa](k) \vep_{\lambda}(k) \cdot \mathcal{F}[\vj_t](k),    \\
\boldsymbol{A}(x,t) &=  (2 \pi)^{-3/2}   \sum_{\lambda=1,2} \int d^3k \, \frac{1}{\sqrt{2 \abs{k}}} \vep_{\lambda}(k) 
\left( e^{ikx} \alpha_t(k,\lambda) + e^{-ikx} \overline{\alpha_t(k,\lambda)}  \right) ,
\\
\vj_t &= 2 \left(  \Im(\varphi_t^* \nabla \varphi_t) - \abs{\varphi_t}^2  \kappa * \boldsymbol{A}(\cdot,t) \right)
\end{cases}
\end{align}
with initial datum $(\varphi_0,\alpha_0(\cdot,1), \alpha_0(\cdot,2))$.
As it will be clarified shortly, the latter appear naturally as the effective counterparts of the microscopic dynamical variables.
Note that the energy functional of the Maxwell-Schr\"odinger system is given by
\begin{align}
\label{eq: Pauli energy functional of the HM system}
\mathcal{E}_M\left[\varphi, \alpha \right]
&\coloneqq \norm{\left(- i \nabla -  (\kappa * \boldsymbol{A}) \right) \varphi}^2 + \frac{1}{2} \scp{\varphi}{\left( v * \abs{\varphi}^2 \right) \varphi}  
+ \sum_{\lambda=1,2} \int d^3 k \abs{k} \abs{\alpha(k,\lambda)}^2  
\end{align}
with $\boldsymbol{A}$ being defined in analogy to \eqref{eq: complex-real fields}.
Global well-posedness for the Maxwell-Schr\"odinger system with
$\mathcal{V}[\varphi ]= v* \lvert \varphi \rvert_{}^{2}\varphi$, $\kappa(x)=\mathrm{e} \,\delta(x)$ (where
$\mathrm{e}$ is the electric charge of the Schr\"odinger particle), and
$v(x)=\frac{\mathrm{e}^2}{\lvert x \rvert_{}^{}}$ has been proven in
\cite{bejenaru,nakamurawada}. We will also consider only the case
$\mathcal{V}[\varphi]= v* \lvert \varphi \rvert_{}^{2}\varphi$, but we may require the charge
distribution $\kappa$ to be extended, in order to well-define the microscopic
system, as discussed below. Typical examples of charge distributions that we
will consider are of the form
\begin{equation}
  \kappa(x)=  \frac{\mathrm{e}}{\mathrm{\sigma}^3 (2\pi)^{3/2}}e^{-\frac{x^2}{2\mathrm{\sigma}^2}}\;,
\end{equation}
representing a charged particle with total charge $\mathrm{e}\in \mathbb{R}$,
distributed in a Gaussian fashion (``smoothed'' spherical distribution of
``diameter'' $\sigma$); or
\begin{equation}
\label{eq:definition charge distribution with sharp cutoff}
  \kappa(x)= \mathrm{e}\frac{\mathcal{F}[\id_{\lvert \cdot   \rvert_{}^{}\leq \Lambda}](x)}{(2\pi)^{3/2}}\; ,
\end{equation}
where $\mathcal{F}[\,\cdot \,]$ stands for the Fourier transform, representing a
sharp cutoff in momentum space, with total charge $\mathrm{e}\in \mathbb{R}$. Let us
remark that globally neutral particles can be considered, as long as they
have a nontrivial charge distribution: for example,
\begin{equation} 
  \kappa(x)= 
  \begin{cases}
    \frac{e^{-\frac{x^2}{2}}}{(2\pi)^{3/2}} & \text{if } x=(x_1,x_2,x_3)\;,\; x_1\geq 0\\
    -\frac{e^{-\frac{x^2}{2}}}{(2\pi)^{3/2}} & \text{if } x=(x_1,x_2,x_3)\;,\; x_1< 0
  \end{cases}\; ,
\end{equation}
yields null total charge but nontrivial dipole, quadrupole, etc.\
interactions with the electromagnetic field.

If the charge distribution is not concentrated in a single point, and the
potential $v$ represents an electrostatic mean-field self-interaction, then
the form of the latter changes as well: a physically sensible choice would be
$v=\kappa*\frac{1}{\lvert\, \cdot \,\rvert_{}^{}}*\kappa$. We will allow some liberty in the choices of
$\kappa$ and $v$; the specific requirements on the two will be made precise in Assumption \ref{assumption:potential and charge distribution} below. Concerning global well-posedness, let us remark that
compared to the literature \cite{bejenaru,nakamurawada} our choices for $\kappa$ and $v$ will be, at most, ``better'' (\textit{i.e.}, more regular) and
therefore do not affect the proof in any way since both $\kappa$ and $v$ act by
convolution in the equation.
\begin{proposition}[\cite{bejenaru}]
 \label{prop:1}
 The Maxwell-Schr\"odinger system in Coulomb's gauge \eqref{eq:Hartree-Maxwell 2} (with variables $(\varphi,\mathbf{A},\mathbf{E})$) is
  globally well-posed in $H^1\times H^1\times L^2$. More precisely,
  \begin{enumerate}
  \item\label{item:1} (Regular solutions -- \cite{nakamurawada}) For every
    \begin{equation*}
      (\varphi_0,\mathbf{A}_0,\mathbf{E}_0)\in H^2\times H^2\times H^1\;,
    \end{equation*}
    there exists a unique global solution
    \begin{equation*}
      (\varphi,\mathbf{A})\in \big( \mathscr{C}^0(\mathbb{R},H^2)\cap \mathscr{C}^1(\mathbb{R},L^2) \big) \times \big( \mathscr{C}^0(\mathbb{R},H^2)\cap \mathscr{C}^1(\mathbb{R},H^1)\cap \mathscr{C}^2(\mathbb{R},L^2) \big)
    \end{equation*}
    of the Cauchy problem associated to \eqref{eq:Hartree-Maxwell 2}.
    
  \item\label{item:2} (Rough solutions) For every
    \begin{equation*}
      (\varphi_0,\mathbf{A}_0,\mathbf{E}_0)\in H^1\times H^1\times L^2\;,
    \end{equation*}
    there exists a unique global solution
    \begin{equation*}
      (\varphi,\mathbf{A})\in \mathscr{C}^0(\mathbb{R},H^1)\times \big( \mathscr{C}^0(\mathbb{R},H^1)\cap \mathscr{C}^1(\mathbb{R},L^2) \big)
    \end{equation*}
 of the Cauchy problem associated to \eqref{eq:Hartree-Maxwell 2}, being the unique strong limit of a sequence of regular solutions in \eqref{item:1}, whose initial data approximate the rough initial datum in $H^1\times H^1\times L^2$.
\item\label{item:3} (Continuous dependence on initial data) The solutions $(\varphi,\mathbf{A})$ in \eqref{item:2} depend continuously on the initial datum $(\varphi_0,\mathbf{A}_0,\mathbf{E}_0)\in H^1\times H^1\times L^2$.
  \end{enumerate}
\end{proposition}

For $m \in \mathbb{R}$ , let $\mathfrak{h}_{m}$ denote the weighed $L^2(\mathbb{R}^3) \otimes \mathbb{C}^2$-space with norm 
\begin{align}
\norm{\alpha}_{\mathfrak{h}_m} = \Big( \sum_{\lambda=1,2} \int d^3k \, \left( 1 + \abs{k}^2 \right)^m \abs{\alpha(k,\lambda)}^2 \Big)^{1/2} .
\end{align}
Throughout this work we will rely on the following statement which results almost immediately from Proposition \ref{prop:1} (see Appendix \ref{section:properties of the solutions of Hatree-Maxwell}).
\begin{corollary}
\label{corollary:Maxwell-Schroedinger mode function}
Let $\abs{\cdot}^{-1/2} \mathcal{F}[\kappa] \in L^2(\mathbb{R}^3, \mathbb{C})$. For every initial datum $(\varphi_0, \alpha_0) \in H^2 (\mathbb{R}^3, \mathbb{C}) \times \big( \mathfrak{h}_{\frac{3}{2}} \cap \mathfrak{h}_{- \frac{1}{2}} \big)$ the Maxwell-Schr\"odinger system \eqref{eq:Hatree-Maxwell}  has a unique global solution in $H^2(\mathbb{R}^3, \mathbb{C}) \times  \mathfrak{h}_{\frac{3}{2}}$.
\end{corollary}

\subsection{The Microscopic Model: Pauli-Fierz Hamiltonian}
\label{sec:micr-model:-pauli}

The microscopic model corresponding to the Maxwell-Schr\"odinger system with
mean-field self-interaction $v*\lvert \varphi \rvert_{}^2\varphi$ consists of many identical
nonrelativistic particles -- obeying Bose-Einstein condensation -- interacting
among themselves by means of a weak pair potential and with a quantized
electromagnetic field in Coulomb's gauge. Contrarily to the ``classical''
case, the microscopic model is known to be well-defined only for extended
charges.
Let us start by defining a Hilbert space $\mathcal{H}^{(N)}_{\hslash}$
depending on two parameters $N\in \mathbb{N},\hslash\in \mathbb{R}^+$ as follows:
\begin{equation}
  \label{eq:hilbertspacemicro}
\mathcal{H}^{(N)}_{\hslash} := L^2_{\mathrm{s}}(\mathbb{R}^{3N})\otimes \Gamma_{\hslash}\bigl(L^2(\mathbb{R}^3)\otimes \mathbb{C}^2\bigr)\;,
\end{equation}
where $L^2_{\mathrm{s}}(\mathbb{R}^{3N})$ is the natural Hilbert space of $N$
identical bosons (the subscript $_{\mathrm{s}}$ indicates symmetry under the interchange of variables) and $\Gamma_{\hslash}$ is the second quantization functor associating
to any (pre-)Hilbert space $\mathfrak{h} = L^2(\mathbb{R}^3) \otimes \mathbb{C}^2$ the corresponding Fock
representation of the Canonical Commutation Relations
$
  [a_{\hslash}(f),a^{*}_{\hslash}(g)]=\hslash\langle f  , g \rangle_{\mathfrak{h}}\; ,
$
with $\hslash$ a semiclassical parameter measuring the degree of noncommutativity
of the quantum field. The Fock representation is the natural one to describe
noninteracting or regularized quantum field theories, the latter being the
case here with $\mathfrak{h}:= L^2(\mathbb{R}^3)\otimes \mathbb{C}^2$. With this interpretation, the
limits $N\to \infty$ and $\hslash\to 0$ describe respectively the regimes in which the
bosons are many and the quantum effects of the field are negligible.
The time evolution is dictated by the Schr\"odinger equation
\begin{equation}
\label{eq:Schroedinger equation microscopic}
  i\partial_t \Psi_{N,\hslash}(t)= H_{N,\hslash}\Psi_{N,\hslash}(t)\; , 
\end{equation}
where the Hamiltonian $H_{N,\hslash}$, called Pauli-Fierz Hamiltonian, is given by
\begin{equation}
  \label{eq:pfhNhregime}
  H_{N,\hslash}= \sum_{j=1}^N\bigl(-i\nabla_j-\mu_{N,\hslash}\,\hat{\mathbf{A}}_{\kappa}(x_j)\bigr)^2 + g_N \sum_{1\leq j<k\leq N}^{} v(x_j-x_k) + \frac{1}{\hslash} H_f \; ,
\end{equation}
where $\mu_{N,\hslash}$ describes the coupling strength between the particles and the
field, $g_N$ the coupling strength between the particles, $\kappa$ and $v$ are the charge distribution and the pair potential introduced
previously,
\begin{equation}
H_f =\sum_{\lambda=1,2} \int \mathrm{d}^3k \, \abs{k} a^*_{\hslash}(k,\lambda) a_{\hslash}(k,\lambda)
\end{equation}
is the field's kinetic energy, with $a^{\sharp}_{\hslash}(k,\lambda)$ the polarized creation and
annihilation operators satisfying the CCR
\begin{equation}
  [a_{\hslash}(k,\lambda),a_{\hslash}^{*}(k',\lambda')]= \hslash \delta_{\lambda\lambda'}\delta(k-k')\;,\; [a_{\hslash}(k,\lambda),a_{\hslash}(k',\lambda')]= [a^{*}_{\hslash}(k,\lambda),a^{*}_{\hslash}(k',\lambda')]=0\; ,
\end{equation}
and\footnote{To simplify the notation we assume $\mathcal{F}[\kappa](k) \in \mathbb{R}$ for all $k \in \mathbb{R}^3$. Theorem \ref{theorem: Pauli main theorem} equally applies if $\mathcal{F}[\kappa]$ is complex valued. In this case,
$\hat{\mathbf{A}}_{\kappa}(x)= \sum_{\lambda=1,2} \int \mathrm{d}^3k \,   \frac{1}{\sqrt{2 \abs{k}}} \vep_{\lambda}(k) \left( \overline{\mathcal{F}[\kappa](k)}  e^{ikx} a_{\hslash}(k,\lambda) + \mathcal{F}[\kappa](k) e^{-ikx} a_{\hslash}^*(k,\lambda)  \right) $.
}
\begin{equation}
  \hat{\mathbf{A}}_{\kappa}(x)= \sum_{\lambda=1,2} \int \mathrm{d}^3k \,   \frac{\mathcal{F}[\kappa](k)}{\sqrt{2 \abs{k}}} \vep_{\lambda}(k) \left( e^{ikx} a_{\hslash}(k,\lambda) + e^{-ikx} a_{\hslash}^*(k,\lambda)  \right)
\end{equation}
the smeared quantized electromagnetic vector potential in Coulomb's
gauge. Let us remark that both $\hat{\mathbf{A}}_{\kappa}(x)$ and $H_f$ depend on $\hslash$ through the creation and annihilation operators,
that have $\hslash$-dependent CCRs. The Hamiltonian $H_{N,\hslash}$ is self-adjoint on
$\mathcal{D}(H_{N,\hslash})= \mathcal{D}(H_{N,\hslash}^{(0)})$, where $H_{N,\hslash}^{(0)}=
H_{N,\hslash}\Bigr\rvert_{\mu_{N,\hslash}=g_N=0}$, whenever $\left(  \abs{\cdot}^{-1} + \abs{\cdot}^{1/2} \right) \mathcal{F}[\kappa] \in L^2(\mathbb{R}^3)$
and $v$ is Kato-infinitesimal with respect to $-\Delta$ \cite{hiroshima,matte, S2004}.

\subsection{Scaling regime}

Our aim is to prove that the Maxwell-Schr\"odinger system emerges in some limit
$N\to \infty$ and/or $\hslash\to 0$, as an effective model of the microscopic Pauli-Fierz
dynamics. This is true only if we couple the parameters $N,\hslash$ suitably, and
choose the coupling constants $\mu_{N,\hslash}, g_N$ accordingly. A possible choice
is given by $N\to \infty$, $\hslash=\frac{1}{N}$, $\mu_{N,\hslash}=1$, $g_N=\frac{1}{N}$. In this
regime the electromagnetic field becomes classical inverse proportionally to
the increasing number of bosons. At the same time the coupling between
the particles and field is of order one, while the coupling between pairs of
particles becomes weak (of order $\frac{1}{N}$). A mathematically equivalent
but physically different choice is given by $N\to \infty$, $\hslash=1$,
$\mu_{N,\hslash}=\frac{1}{\sqrt{N}}$, $g_N=\frac{1}{N}$. Here, the physical interpretation is of many bosons that interact weakly both with the quantized
electromagnetic field (coupling of order $\frac{1}{\sqrt{N}}$) and among
themselves (pair coupling of order $\frac{1}{N}$).
Our result reads as follows:
 
\textit{
Provided that we choose an initial microscopic
state that is ``close enough'' to a non-interacting state representing a
complete condensate and a coherent field of minimal uncertainty, then at any
time $t \geq 0$ the evolution keeps the state ``close'' (in the same sense as
above) to an analogous configuration in which the one-particle wave function
and the argument of the coherent field have been evolved by the coupled
Maxwell-Schr\"odinger equations.}

The described scaling regime has been considered in earlier works for the Nelson model with ultraviolet cutoff \cite{AF2014, F2013, FLMP2021, LP2018}, the renormalized Nelson model \cite{AF2017} and the Fr\"ohlich model \cite{LMS2021}.  In \cite{LP2019} the Nelson model with ultraviolet cutoff has been studied in a limit of many weakly interacting fermions.  The classical behavior of quantum fields has also been proven in different scaling regimes \cite{AFH2022, CCFO2019, CF2018, CFO2019, CFO2020, D1979, FRS2021, FG2017, FS2014, GNV2006, G2017, K2009, LMRSS2021, LRSS2019, M2021, T2002}.
We also would like to mention \cite{S2010} which derives the Maxwell-Schr\"odinger equations in a nonrigorous manner by neglecting certain terms in the Pauli-Fierz Hamiltonian.

\section{Main Result}
\label{sec:main-result:-schr}

From now on, we will keep $N$ as the single parameter and choose $\hslash=1$,
$\mu_{N,\hslash}=\frac{1}{\sqrt{N}}$, $g_N=\frac{1}{N}$. We will use the notations $\mathcal{H}^{(N)} = \mathcal{H}^{(N)}_{\hslash}$, $\mathcal{F}_p = \Gamma_{1}\bigl(L^2(\mathbb{R}^3)\otimes \mathbb{C}^2\bigr)$ with vacuum $\Omega$, $H_N = H_{N,1}$ and $\Psi_N = \Psi_{N,1}$. 
Concerning the interaction potential and charge distribution we will make the assumptions.
\begin{assumption}
\label{assumption:potential and charge distribution}
The (repulsive) interaction potential $v$ is a positive, real, and even function satisfying 
\begin{align}
\norm{v}_{L^2+L^{\infty}(\mathbb{R}^3)} = 
\inf_{v=v_1 + v_2} \{ \norm{v_1}_{L^2(\mathbb{R}^3)}
+ \norm{v_2}_{L^{\infty}(\mathbb{R}^3)}  \} < + \infty.
\end{align}
The charge distribution $\kappa$ with Fourier transform $\mathcal{F}[\kappa]$ satisfies
\begin{equation}
\label{eq: Pauli cut off function}
\big( \abs{\cdot}^{-1} +  \abs{\cdot}^{1/2} \big) \mathcal{F}[\kappa] \in L^2(\mathbb{R}^3)  .
\end{equation}
\end{assumption}
In order to state our result we define for $\Psi_N \in \mathcal{H}^{(N)}$ the one-particle reduced density matrix of the charged particles $\gamma_{\Psi_{N}}: L^2(\mathbb{R}^3) \rightarrow L^2(\mathbb{R}^3)$ by
\begin{align}
\label{eq: definition reduced one-particle matrix charged particles}
\gamma_{\Psi_N}^{(1,0)} \coloneqq \tr_{2,\ldots, N} \, \tr_{\mathcal{F}_p} \ket{\Psi_{N}} \bra{\Psi_{N}},
\end{align}
where $\tr_{2,\ldots, N}$  denotes the partial trace over the coordinates $x_2,\ldots, x_N$ and $\tr_{\mathcal{F}_p}$ is the trace over Fock space. In addition, we introduce the number of photon operator
\begin{align}
\mathcal{N} = \sum_{\lambda = 1,2} \int d^3 k \, a^*(k,\lambda) a(k,\lambda)
\end{align}
and the unitary Weyl operator 
\begin{align}
W(f) &= 
\exp \Big( \sum_{\lambda = 1,2} \int d^3k \, f(k,\lambda) a^*(k,\lambda) - \overline{f(k, \lambda)} a(k, \lambda)  \Big)
\quad \text{with} \; f \in \mathfrak{h}.
\end{align}
Our result is the following.

\begin{theorem}
\label{theorem: Pauli main theorem}
Let $v$ and $\kappa$ satisfy Assumption \ref{assumption:potential and charge distribution},  $(\varphi_0, \alpha_0) \in H^2(\mathbb{R}^3, \mathbb{C}) \times \big( \mathfrak{h}_{\frac{3}{2}} \cap \mathfrak{h}_{- \frac{1}{2}} \big)$ with $\norm{\varphi_0}_{L^2(\mathbb{R}^3)} = 1$
and $\Psi_{N,0} \in  \mathcal{D}\left( H_N \right) \cap \mathcal{D} \left( \mathcal{N}^{1/2} \right)$ such that $\norm{\Psi_{N,0}}_{\mathcal{H}^{(N)}} = 1$. Define
\begin{align}
a_N &\coloneqq \tr_{L^2(\mathbb{R}^3)} \abs{\gamma_{\Psi_{N,0}}^{(1,0)} - \ket{\varphi_0} \bra{\varphi_0}} 
\\
\label{eq:initial condition in main theorem coherence of photons}
b_N &\coloneqq N^{-1} \scp{W^{-1}(\sqrt{N} \alpha_0)\Psi_{N,0}}{\mathcal{N} W^{-1}(\sqrt{N} \alpha_0)\Psi_{N,0}}_{\mathcal{H}^{(N)}}
\quad \text{and}  \\
\label{eq:initial condition in main theorem variance of the energy}
c_N  &\coloneqq \norm{\left(N^{-1} H_N - \mathcal{E}_M\left[\varphi_0, \alpha_0  \right] \right) \Psi_{N,0} }_{\mathcal{H}^{(N)}}^2  .
\end{align}
Let $(\varphi_t, \alpha_t)$ and  $\Psi_{N,t}$ be the unique solutions of \eqref{eq:Hatree-Maxwell} and \eqref{eq:Schroedinger equation microscopic}, respectively.
Then, there exists a monotone increasing function $C(s)$ of the norms $\norm{\varphi_s}_{H^2(\mathbb{R}^3)}$, $\norm{\abs{\cdot}^{1/2} \alpha_s}_{\mathfrak{h}}$, $\norm{v}_{L^2 + L^{\infty}(\mathbb{R}^3)}$ and 
$\norm{\left( \abs{\cdot}^{-1/2} + \abs{\cdot}^{-1} \right) \mathcal{F}[\kappa]}_{L^2(\mathbb{R}^3)}$
such that 
\begin{align}
 \label{eq: main theorem 1}
\tr_{L^2(\mathbb{R}^3)} \abs{\gamma_{N,t}^{(1,0)} - \ket{\varphi_t} \bra{\varphi_t}} &\leq
\sqrt{a_N + b_N + c_N + N^{-1} } \,  e^{ \int_0^t ds \, C(s)} ,  \\
\label{eq: main theorem 2}
N^{-1} \scp{W^{-1}(\sqrt{N} \alpha_t)\Psi_{N,t}}{\mathcal{N} W^{-1}(\sqrt{N} \alpha_t)\Psi_{N,t}}_{\mathcal{H}^{(N)}}  &\leq  
\left( a_N + b_N + c_N + N^{-1} \right) \,   e^{\int_0^t ds \, C(s)}  
\end{align}
for any $t \geq 0$.
In particular, for $\Psi_{N,0} = \varphi_{0}^{\otimes N} \otimes  W(\sqrt{N} \alpha_0) \Omega$ one obtains
\begin{align}
\label{eq: main theorem 3}
\tr_{L^2(\mathbb{R}^3)} \abs{\gamma_{N,t}^{(1,0)} - \ket{\varphi_t} \bra{\varphi_t}} &\leq
 N^{-1/2}  C(0) e^{\int_0^t ds \, C(s)},  \\
\label{eq: main theorem 4}
N^{-1} \scp{W^{-1}(\sqrt{N} \alpha_t)\Psi_{N,t}}{\mathcal{N} W^{-1}(\sqrt{N} \alpha_t)\Psi_{N,t}}_{\mathcal{H}^{(N)}}  &\leq
 N^{-1}  C(0) e^{\int_0^t ds \, C(s)}. 
\end{align}
\end{theorem}

\begin{remark}
Let $\gamma_{\Psi_{N,t}}^{(0,1)}$ be the one-particle reduced density matrix of the photons with integral kernel
\begin{align}
\gamma_{\Psi_{N,t}}^{(0,1)}(k, \lambda; k' , \lambda') = N^{-1} \scp{\Psi_{N,t}}{a^*(k', \lambda') a(k,\lambda) \Psi_{N,t}}_{\mathcal{H}^{(N)}} .
\end{align}
By similar means as in \cite[Lemma 5.3]{LP2020} one obtains
\begin{align}
\tr_{\mathfrak{h}} \abs{\gamma_{N,t}^{(0,1)} - \ket{\alpha_t} \bra{\alpha_t}} &\leq \max_{j=1,2} \left( a_N + b_N + c_N + N^{-1} \right)^{j/2} \, \left( 1 + \norm{\alpha_t}_{\mathfrak{h}} \right)   e^{\int_0^t ds \, C(s)}  
\end{align}
from \eqref{eq: main theorem 2} and
\begin{align}
\tr_{\mathfrak{h}} \abs{\gamma_{N,t}^{(0,1)} - \ket{\alpha_t} \bra{\alpha_t}} &\leq N^{-1/2} \left( 1 + \norm{\alpha_t}_{\mathfrak{h}} \right)   C(0) e^{\int_0^t ds \, C(s)}  
\end{align}
for initial product states $\Psi_{N,0} = \varphi_{0}^{\otimes N} \otimes  W(\sqrt{N} \alpha_0) \Omega$ from \eqref{eq: main theorem 4}.
\end{remark}

\begin{remark}
In \cite{LP2020} Theorem \ref{theorem: Pauli main theorem} was proven for the charge distribution \eqref{eq:definition charge distribution with sharp cutoff} and with the number operator $\mathcal{N}$ in \eqref{eq:initial condition in main theorem coherence of photons}, \eqref{eq: main theorem 2} and \eqref{eq: main theorem 4} being replaced by the field energy $H_f$. Because of Markov's inequality, 
\begin{align*}
\sum_{\lambda = 1,2} \int_{\abs{k} \geq I} d^3 k \,  a^*(k,\lambda) a(k, \lambda) \leq I^{-1} H_f ,
\end{align*}
one can use the field energy to conclude that the quantum fluctuations around the coherent state are subleading for all photons with $\abs{k} \geq I$. For sufficiently small $a_N$, $b_N$ and $c_N$ one can choose $I \sim N^{-a}$ with $a < 1$. However, this choice does not provide information about the coherence of soft photons with frequencies below this threshold.
\end{remark}

\section{Proof of the result}

The rest of the article outlines the proof of Theorem \eqref{theorem: Pauli main theorem}. We will proceed as follows:
\begin{enumerate}

\item We define a functional $\beta[\Psi_N, \varphi, \alpha]$ which measures if the charges of the many-body state $\Psi_N$ form a Bose-Einstein condensate with condensate wave function $\varphi$ and if the photons are in a coherent state with mean photon number $N \norm{\alpha}_{\mathfrak{h}}^2$.

\item Next, we show that the domain of $\beta$ contains the solutions $(\varphi_t, \alpha_t)$ of \eqref{eq:Hatree-Maxwell} and $\Psi_{N,t}$ of \eqref{eq:Schroedinger equation microscopic} from Theorem \eqref{theorem: Pauli main theorem} for all $t \geq 0$.

\item  Afterwards, we compute the change of $\beta[\Psi_{N,t}, \varphi_t, \alpha_t]$ in time.

\item Finally, we control the growth of $\beta[\Psi_{N,t}, \varphi_t, \alpha_t]$ with the help of Gr\"onwall's inequality. This concludes the proof.

\end{enumerate}

In doing so, we will rely on the findings from \cite{LP2020} and rather explain how the original proof of \cite{LP2020} has to be adapted. Most of the modifications are necessary to show the invariance of the domain in Step 2 and to compute the time derivative of $\beta$ in Step 3.

\subsection{Definition of the functional}

We define a functional which consists of three parts.
\begin{definition}
\label{definition:functional}
For $\varphi \in L^2(\mathbb{R}^3)$ we define $p_1^{\varphi}: L^2(\mathbb{R}^{3N}) \rightarrow L^2(\mathbb{R}^{3N})$ by
\begin{align}
p_1^{\varphi} f(x_1, \ldots, x_N)
&\coloneqq \varphi(x_1) \int d^3 x_1 \overline{\varphi(x_1)} f(x_1, \ldots, x_N)
\end{align}
and $q_1^{\varphi} \coloneqq 1_{L^2(\mathbb{R}^{3N})} - p_1^{\varphi}$.
Now, let $\Psi_{N} \in \mathcal{D}(H_N) \cap \mathcal{D} \left( \mathcal{N}^{1/2}  \right)$, $\varphi \in H^1(\mathbb{R}^3)$, $\alpha \in \mathfrak{h}_{\frac{1}{2}}$. Then,
\begin{align}
\begin{split}
\beta^a(\Psi_{N},\varphi) &\coloneqq  \SCP{\Psi_{N}}{q_1^{\varphi} \otimes \id_{\mathcal{F}_p} \, \Psi_{N}}_{\mathcal{H}^{(N)}} , 
\\
\beta^b(\Psi_{N},\alpha) &\coloneqq 
 \sum_{\lambda=1,2} \int d^3k \,  \SCP{\left( \frac{a(k,\lambda)}{\sqrt{N}} - \alpha(k,\lambda) \right) \Psi_{N}}{ \left( \frac{a(k,\lambda)}{\sqrt{N}} - \alpha(k,\lambda) \right) \Psi_{N}}_{\mathcal{H}^{(N)}} ,
\\
\beta^c(\Psi_{N},\varphi, \alpha) &\coloneqq \SCP{\left( \frac{H_N}{N} - \mathcal{E}_M[\varphi, \alpha] \right) \Psi_{N}}{\left( \frac{H_N}{N} - \mathcal{E}_M[\varphi, \alpha] \right) \Psi_{N}}_{\mathcal{H}^{(N)}}
\end{split}
\end{align}
and the functional $\beta: \big( \mathcal{D}(H_N) \cap \mathcal{D} \left( \mathcal{N}^{1/2} \right) \big) \times H^1(\mathbb{R}^3) \times \mathfrak{h}_{\frac{1}{2}}  \rightarrow \mathbb{R}_0^+ $
is defined as
$ \beta \coloneqq \beta^a + \beta^b + \beta^c$.
\end{definition}
The functional $\beta^a$ measures if the charges of the many-body state are in a Bose-Einstein condensate (we refer to \cite{KP2009, P2011} for a comprehensive introduction). Its relation to the trace norm distance of the one-particle reduced density matrix is given by (see, e.g. \cite[Lemma 5.3]{LP2020})
\begin{align}
\label{eq:relation beta-a and reduced density}
\beta^a(\Psi_{N},\varphi) \leq \tr_{L^2(\mathbb{R}^3)} \abs{\gamma_{\Psi_{N,0}}^{(1,0)}  - \ket{\varphi} \bra{\varphi}} \leq \sqrt{8 \beta^a(\Psi_{N},\varphi)} .
\end{align}
The functional $\beta^b$ quantifies the fluctuations of $\Psi_N$ around the coherent state $W(\sqrt{N} \alpha) \Omega$.  Using property \eqref{eq:shifting property Weyl operators} of the Weyl operators it can be written as
\begin{align}
\label{eq:relation beta-b and number operator}
\beta^b(\Psi_{N},\alpha) &=
N^{-1} \scp{W^{-1}(\sqrt{N} \alpha)\Psi_{N}}{\mathcal{N} \, W^{-1}(\sqrt{N} \alpha)\Psi_{N}}_{\mathcal{H}^{(N)}} ,
\end{align}
showing that it is the same quantity as usually considered in the coherent state approach \cite{RS2009}.
While $\beta^a$ and $\beta^b$ measure the deviation of $\Psi_N$ from the product state $\varphi^{\otimes N} \otimes W(\sqrt{N} \alpha) \Omega$ the functional $\beta^c$ is introduced for technical reasons. It quantifies the fluctuations of the many-body energy per particle around the energy of the Maxwell-Schr\"odinger system.  

In the original  proof of \cite{LP2020} the functional $\beta$ was considered with $\beta^b\left( \Psi_N, \alpha \right)$ being replaced by
\begin{align}
\label{eq:definition beta-b-tilde}
\begin{split}
\widetilde{\beta}^b \left( \Psi_N, \alpha \right)
&\coloneqq \sum_{\lambda=1,2} \int d^3k \, \abs{k}  \SCP{\left( \frac{a(k,\lambda)}{\sqrt{N}} - \alpha(k,\lambda) \right) \Psi_{N}}{ \left( \frac{a(k,\lambda)}{\sqrt{N}} - \alpha(k,\lambda) \right) \Psi_{N}}_{\mathcal{H}^{(N)}}
\\
&=
N^{-1} \scp{W^{-1}(\sqrt{N} \alpha)\Psi_{N}}{H_f \, W^{-1}(\sqrt{N} \alpha)\Psi_{N}}_{\mathcal{H}^{(N)}} .
\end{split}
\end{align}
This definition has the advantage that it can be defined for many-body states $\Psi_N$ in the domain $\mathcal{D}\left( H_N \right) = \left( H^2(\mathbb{R}^{3N}, \mathbb{C}) \otimes \mathcal{F} \right) \cap \mathcal{D} \left( H_f \right)$ which is invariant under the time evolution $e^{- i H_N t}$. The additional difficulties with respect to \cite{LP2020} actually originate from the fact that $\mathcal{D} \left( \mathcal{N}^{1/2} \right)$  (in contrast to $\mathcal{D} \left( H_f\right)$) is not contained in the domain of the Pauli-Fierz Hamiltonian. 
On the contrary $\widetilde{\beta}^b$ does not allow to investigate the coherence of photons with small frequencies because the factor $\abs{k}$ in the integral on the right hand side of \eqref{eq:definition beta-b-tilde} suppresses contributions from photons with small energies.

\subsection{Invariance of the domain}

Throughout the rest of the article $(\varphi_t, \alpha_t)$ and  $\Psi_{N,t}$ denote the solutions of \eqref{eq:Hatree-Maxwell} and \eqref{eq:Schroedinger equation microscopic} from Theorem \eqref{theorem: Pauli main theorem}. In this section we show that $(\Psi_{N,t}, \varphi_t, \alpha_t) \in \big( \mathcal{D}(H_N) \cap \mathcal{D} \left( \mathcal{N}^{1/2}  \right) \big) \times H^2(\mathbb{R}^3) \times \mathfrak{h}_{\frac{3}{2}}$ for all $t \geq 0$. The condition on the Maxwell-Schr\"odinger solutions is satisfied because of Corollary \ref{corollary:Maxwell-Schroedinger mode function}. While $\mathcal{D} \left( H_N \right)$ is invariant under the evolution of the Pauli-Fierz Hamiltonian, due to Stone's theorem, the invariance of $\mathcal{D} \left( \mathcal{N}^{1/2} \right)$ is less clear because the photon number is not conserved during the time evolution.
The next statement, however, displays that the number of photons can be controlled by the energy of the system.\footnote{Inequality \eqref{eq: bound number operator} was originally proven by Fumio Hiroshima and appeared in a slightly different form (for the second instead of the first moment of the number operator) in \cite[Proposition 3.11]{AFH2022}. We would like to thank Fumio Hiroshima for sharing his notes with us. The proof is presented again for the convenience of the reader.}
\begin{lemma}
\label{lemma: bound number operator}
Let $\Psi_{N,0} \in \mathcal{D} \big( H_N^{1/2} \big) \cap \mathcal{D} \left( \mathcal{N}^{1/2} \right)$. Then, there exists a constant (depending on $N$ and the choice of $\kappa$) such that
\begin{align}
\label{eq: bound number operator}
\norm{\mathcal{N}^{1/2} e^{- i H_N t} \Psi_{N,0}}
&\leq  \norm{\mathcal{N}^{1/2} \Psi_{N,0}} + C \, t^{1/2} \norm{\left( H_N + C \right)^{1/2} \Psi_{N,0}}
\quad \text{for all} \;
t \geq 0 .
\end{align}
This implies $e^{- i H_N t} \, \mathcal{D} \big( H_N^{1/2} \big) \cap \mathcal{D} \left( \mathcal{N}^{1/2} \right) =  \mathcal{D} \big( H_N^{1/2} \big) \cap \mathcal{D} \left( \mathcal{N}^{1/2} \right)$.
\end{lemma}

\begin{proof}[Proof of Lemma \ref{lemma: bound number operator}]

In the following, we use the notation
\begin{align}
a(g) &= \sum_{\lambda=1,2} \int d^3k \, \overline{g(k,\lambda)} a(k,\lambda) ,
\quad 
a^*(g) = \sum_{\lambda=1,2} \int d^3k \, g(k, \lambda) a^*(k,\lambda)
\end{align}
and 
\begin{align}
G_x(k,\lambda) = \mathcal{F}[\kappa](k) \frac{1}{\sqrt{2 \abs{k}}} \vep_{\lambda}(k) e^{-ikx} .
\end{align}
The vector potential can then be written as
$\vA(x) = a \left( G_x \right) + a^* \left( G_x \right)$.
Recall the standard estimates for the annihilation and creation operators
\begin{align}
\begin{split}
\label{eq: bounds for the annihilation and creation operator}
\norm{a(g) \Psi} &\leq \norm{g}_{\mathfrak{h}} \norm{\mathcal{N}^{1/2} \Psi} ,
\\
\norm{a^*(g) \Psi} &\leq \norm{g}_{\mathfrak{h}} \norm{\left( \mathcal{N}  + 1 \right)^{1/2} \Psi},
\\
\norm{a(g) \Psi} &\leq \norm{\abs{\cdot}^{-1/2} g}_{\mathfrak{h}} \norm{H_f^{1/2} \Psi} ,
\\
\norm{a^*(g) \Psi} &\leq \norm{\left( 1 + \abs{\cdot}^{-1/2} \right) g}_{\mathfrak{h}} \norm{\left( H_f  + 1 \right)^{1/2} \Psi} .
\end{split}
\end{align}
Let $\Psi_{N,0} \in   \mathcal{D} \big( H_N^{1/2} \big) \cap \mathcal{D} \left( \mathcal{N}^{1/2} \right)$, $\Psi_{N,t}= e^{- i H_N t} \Psi_{N,0}$, $\delta \geq 0$ and consider the bounded operator $\ND = \mathcal{N} e^{- \delta \mathcal{N}}$. Using
\begin{align}
\label{eq: commutator regularized number operator and Pauli-Fierz Hamiltonian}
\left[ \ND , H_N \right]
&= 2 \sum_{j=1}^N    \left[ \ND , N^{-1/2} \vA(x_j) \right] 
\cdot i \nabla_j 
+ N^{-1}  \sum_{j=1}^N   \left[ \ND ,\vA(x_j) \right] \vA(x_j) 
\nonumber \\
&\quad 
+ N^{-1}  \sum_{j=1}^N  \vA(x_j) \left[ \ND ,\vA(x_j) \right] 
\end{align}
and the Cauchy-Schwarz inequality we estimate
\begin{align}
\frac{d}{dt} \norm{\ND^{1/2} \Psi_{N,t}}^2 
&= i \scp{\ND \Psi_{N,t}}{\left[H_N, \ND \right] \Psi_{N,t}}
\nonumber \\
&\leq C N^{-1/2} \sum_{j=1}^N \norm{\left[ \ND ,\vA(x_j) \right] \Psi_{N,t}}
\left( \norm{i \nabla_j \Psi_{N,t}} 
+ N^{-1/2}  \norm{\vA(x_j) \Psi_{N,t}} \right) .
\end{align}
By means of the canonical commutation relations and the shifting property of the number operator we get
\begin{align}
\label{eq: commutator regularized number operator and vector potential explicitely written}
\left[ \ND , \vA(x) \right] 
&=  \left[ \mathcal{N} \left( 1 - e^{- \delta} \right) - e^{- \delta} \right] e^{- \delta \mathcal{N}} a ( G_x)
+ 
\left[ \mathcal{N} \left( e^{- \delta} - 1 \right) + 1 \right] e^{- \delta \left( \mathcal{N} - 1 \right)}  a^* (G_x) .
\end{align}
Using $e^{- \delta \mathcal{N}} \leq 1$, $\mathcal{N}  ( 1 - e^{- \delta} ) e^{- \delta \mathcal{N} } \leq 1$, $\mathcal{N}  (  e^{- \delta} -1 ) e^{- \delta \left( \mathcal{N} -1 \right) } \big|_{\mathcal{N} \geq 1}
\leq 3$ and \eqref{eq: bounds for the annihilation and creation operator} we obtain
\begin{align}
\label{eq: commutator regularized number operator with vector potential}
\norm{\left[ \ND , \vA(x) \right]  \Psi}
&\leq C \norm{\left( \abs{\cdot}^{-1/2} + \abs{\cdot}^{-1} \right) \mathcal{F}[\kappa]}_{L^2(\mathbb{R}^3)}
\norm{\left( H_f + 1 \right)^{1/2} \Psi} .
\end{align}
Hence, 
\begin{align}
\frac{d}{dt} \norm{\ND^{1/2} \Psi_{N,t}}^2 
&\leq 
C \norm{\left( \abs{\cdot}^{-1/2} + \abs{\cdot}^{-1} \right) \mathcal{F}[\kappa]}_{L^2(\mathbb{R}^3)}^2 \norm{\left( H_f + 1 \right)^{1/2} \Psi_{N,t}}^2 
\nonumber \\
&\quad 
+ \scp{\Psi_{N,t}}{ \sum_{j=1}^N  - \Delta_j \Psi_{N,t}}
\nonumber \\
&\leq 
C \left( 1 + \norm{\left( \abs{\cdot}^{-1/2} + \abs{\cdot}^{-1} \right) \mathcal{F}[\kappa]}_{L^2(\mathbb{R}^3)}^2 \right) \norm{\left( H_N^{(0)} + 1 \right)^{1/2} \Psi_{N,t}}^2
\end{align}
with $H_N^{(0)} = - \sum_{j=1}^N \Delta_j + H_f$.  Note that there exists a constant $C(N, \kappa)$ dependent on the number of particles and the  choice of $\kappa$ such that 
$\big\| \big( H_N^{(0)} + 1 \big)^{1/2}  \Psi \big\| \leq C(N, \kappa) \norm{\left( H_N + C \right)^{1/2} \Psi}$
holds for all $\Psi \in \mathcal{D} \big( H_N^{1/2} \big)$.
This fact follows from $\mathcal{D} \big( H_N^{(0)} \big) = \mathcal{D} \left( H_N \right) = \left( H^2(\mathbb{R}^{3N}, \mathbb{C}) \otimes \mathcal{F}_p \right) \cap \mathcal{D} \left( H_f \right)$ and the closed graph theorem \cite[Theorem 1.3 and Corollary 1.4]{hiroshima}. 
In that regard note that $\big( H_N^{(0)} \big)^{1/2} \big( \left( H_N + C \right)^{1/2} + i \big)^{-1}$ is a closed operator.
We consequently obtain
\begin{align}
\norm{\ND^{1/2} \Psi_{N,t}}^2 &\leq 
\norm{\ND^{1/2} \Psi_{N,0}}^2
+ C(N, \kappa) \int_0^t ds \,  \norm{\left( H_N + C \right)^{1/2} \Psi_{N,s}} .
\end{align}
By the spectral theorem and monotone convergence
\begin{align}
\lim_{\delta \rightarrow 0} \norm{\mathcal{N}_{\delta} \Psi}^2
&=    \lim_{\delta \rightarrow 0}  \int_0^{\infty}
\lambda^2 e^{- 2 \delta \lambda} \scp{\Psi}{d E(\lambda) \Psi} 
=   \int_0^{\infty}
\lambda^2 \scp{\Psi}{d E(\lambda) \Psi} 
= \norm{\mathcal{N} \Psi}^2   .
\end{align}
Together with Stone's theorem this shows the claim.
\end{proof}

\subsection{Computing the change of $\beta$ in time}

In order to estimate the emerging correlations between the particles and the photons during the evolution of the system we compute the change of $\beta \left( \Psi_{N,t}, \varphi_t,  \alpha_t \right)$ in time.
\begin{lemma}
\label{lemma:time derivative of beta-b}
Let $(\varphi_t, \alpha_t)$ and  $\Psi_{N,t}$ be the solutions of \eqref{eq:Hatree-Maxwell} and \eqref{eq:Schroedinger equation microscopic} from Theorem \eqref{theorem: Pauli main theorem} and $\beta^b$ be defined as in Definition \ref{definition:functional}. Then 
\begin{align}
\label{eq:time derivative of beta-b}
\begin{split}
&\beta^b \left( \Psi_{N,t}, \alpha_t \right) - \beta^b \left( \Psi_{N,0}, \alpha_0 \right)
\\
&\quad = 
2 \Re  \int_0^t ds \,
\sum_{\lambda = 1,2} \int d^3 k \, 
\Bigg[
\scp{\Psi_{N,s}}{i \sqrt{\frac{4 \pi^3}{\abs{k}}} \mathcal{F}[\kappa] \vep_{\lambda}(k) \overline{\mathcal{F}[\vj_s](k)} \left( \frac{a(k, \lambda)}{\sqrt{N}} - \alpha_s(k, \lambda) \right) \Psi_{N,s}}
\\
&\qquad + 
2 \scp{\Psi_{N,s}}{i \frac{\mathcal{F}[\kappa](k)}{\sqrt{2 \abs{k}}} \vep_{\lambda}(k) e^{i k x_1} \left( i \nabla_1 + N^{-1/2} \vA(x_1) \right) \left( \frac{a(k, \lambda)}{\sqrt{N}} - \alpha_s(k, \lambda) \right) \Psi_{N,s}}
\Bigg] .
\end{split}
\end{align}
\end{lemma}

\begin{proof}[Proof of Lemma \ref{lemma:time derivative of beta-b}]
Let
\begin{align}
\beta_{\delta}^b(t) = N^{-1} \scp{W^{-1}(\sqrt{N} \alpha_t)\Psi_{N,t}}{\mathcal{N}_{\delta} W^{-1}(\sqrt{N} \alpha_t)\Psi_{N,t}}
\end{align}
with $\mathcal{N}_{\delta}$ being defined as in the proof of Lemma \ref{lemma: bound number operator}. Using that $W^{-1}(\sqrt{N} \alpha_t)$ is strongly differentiable in $t$ from $\mathcal{D} \left( \mathcal{N}^{1/2} \right)$ to $\mathcal{H}^{(N)}$ with (see \cite[Lemma 3.1]{GV1979})
\begin{align}
\frac{d}{dt} W^{-1}(\sqrt{N} \alpha_t)
&= \left( N^{1/2} \left( a \left( \dot{\alpha}_t \right) - a^* \left( \dot{\alpha}_t \right) -  N i \, \Im \scp{\alpha_t}{\dot{\alpha}_t} \right) \right) W^{-1}(\sqrt{N} \alpha_t)
\end{align}
and 
\begin{align}
\begin{split}
\label{eq:shifting property Weyl operators}
W^{-1} \left( f \right) a(k,\lambda) W \left( f \right) &= a(k,\lambda) + f(k, \lambda) ,
\\
W^{-1} \left( f \right) a^*(k,\lambda) W \left( f \right) &= a(k,\lambda) + \overline{f(k, \lambda)}
\end{split}
\end{align}
we obtain that the time derivative of the fluctuation vector $\xi_{N,t} = W^{-1}(\sqrt{N} \alpha_t) \Psi_{N,t}$ is given by 
$\frac{d}{dt} \xi_{N,t} = - i \mathcal{G}(t) \xi_{N,t}$ with
\begin{align}
\mathcal{G}(t) &= N^{-1} \sum_{1\leq j<k\leq N} v(x_j-x_k) 
+ H_f + N \big\| \abs{\cdot}^{1/2 }\alpha_t \big\|_{\mathfrak{h}}^2 + N \, \Im \scp{\alpha_t}{\dot{\alpha}_t}_{\mathfrak{h}} 
\nonumber \\
&\quad + N^{1/2} \big( a \left( \abs{\cdot} \alpha_t - i \dot{\alpha}_t \right) + a^* \left( \abs{\cdot} \alpha_t - i \dot{\alpha}_t \right) \big)
\nonumber \\
&\quad + \sum_{j=1}^N  \left( - i \nabla_j - N^{-1/2} \vA(x_j)  - \vAc(x_j,t) \right)^2 .
\end{align}
In analogy to \eqref{eq: commutator regularized number operator and vector potential explicitely written} and the subsequent discussion one shows
\begin{align}
\lim_{\delta \rightarrow 0} \scp{\xi_{N,t}}{\big( \left[ a(f), \mathcal{N}_{\delta} \right] - a(f) \big) \xi_{N,t}} = 0
\end{align}
and
\begin{align}
\lim_{\delta \rightarrow 0} 
\norm{\Big( \left[ \vA(x_j) , \mathcal{N}_{\delta} \right] - \big(  a \left( G_{x_j} \right) - a^* \left( G_{x_j} \right) \big)  \Big) \xi_{N,t}} = 0 
\end{align}
for $f \in \mathfrak{h}$ and $\xi_{N,t} \in \mathcal{D} \left( H_N \right) \cap \mathcal{D} \big( \mathcal{N}^{1/2} \big)$.
Together with 
\begin{align}
\frac{d}{dt} \beta_{\delta}^b(t) 
= i N^{-1} \scp{\xi_{N,t}}{\left[ \mathcal{G}(t) , \mathcal{N}_{\delta} \right] \xi_{N,t}}
\end{align} 
this leads to
\begin{align}
&\lim_{\delta \rightarrow 0} \frac{d}{dt} \beta_{\delta}^b(t) 
\nonumber \\
&\quad = 2  N^{-3/2} \Im \scp{\xi_{N,t}}{\sum_{j=1}^N \left( a \left( G_{x_j}  \right) - a^* \left( G_{x_j} \right) \right)
\left( - i \nabla_j - N^{-1/2} \vA(x_j)  - \vAc(x_j,t) \right) \xi_{N,t}} 
\nonumber \\
&\qquad +  i N^{-1/2} \scp{\xi_{N,t}}{\big( a \left( \abs{\cdot} \alpha_t - i \dot{\alpha}_t \right) - a^* \left( \abs{\cdot} \alpha_t - i \dot{\alpha}_t \right) \big) \xi_{N,t}} .
\end{align}
The claim then follows by Duhamel's formula, monotone convergence and straightforward manipulations using \eqref{eq:shifting property Weyl operators}, $\left[ \nabla_j , \vA(x_j) \right] = 0$ and $\left[ \nabla_j , G_{x_j}(k, \lambda) \right] = 0$.
\end{proof}

From \cite[Section 6.2 and Section 6.4]{LP2020} and Duhamel's formula we immediately obtain the following.

\begin{lemma}
\label{lemma:time derivative of beta-a}
Let $(\varphi_t, \alpha_t)$ and  $\Psi_{N,t}$ be the solutions of \eqref{eq:Hatree-Maxwell} and \eqref{eq:Schroedinger equation microscopic} from Theorem \eqref{theorem: Pauli main theorem} and $\beta^a$, $\beta^c$ be defined as in Definition \ref{definition:functional}. Then, 
\begin{align}
\label{eq:time derivative of beta-a}
\begin{split}
&\beta^a \left( \Psi_{N,t}, \varphi_t \right)
- \beta^a \left( \Psi_{N,0}, \varphi_0 \right)
\\
&\quad = - 2  \int_0^t ds \,
\Bigg[
2 \Re  \scp{\Psi_{N,s}}{p_1^{\varphi_s} \left( N^{-1/2} \vA(x_1) - \vAc(x_1,s) \right) \cdot \nabla_1 q_1^{\varphi_s} \Psi_{N,s}}
\\
&\qquad \qquad  \qquad 
+ \Im \scp{\Psi_{N,s}}{p_1^{\varphi_s} \left( N^{-1} \vA^2(x_1) - \vAc^2(x_1,s) \right) q_1^{\varphi_s} \Psi_{N,s}}
\\
&\qquad \qquad  \qquad 
+ \Im \scp{\Psi_{N,s}}{p_1^{\varphi_s} \left( (N-1) N^{-1} v(x_1 - x_2) - \left( v * \abs{\varphi_t}^2 \right)(x_1) \right) q_1^{\varphi_s} \Psi_{N,s}}
\Bigg]  .
\end{split}
\end{align}
Moreover,
\begin{align}
\beta^c \left( \Psi_{N,t}, \varphi_t, \alpha_t \right)
&=
\beta^c \left( \Psi_{N,0}, \varphi_0, \alpha_0 \right)
\end{align}
due to energy conservation.
\end{lemma}

\subsection{Controlling the growth of $\beta$ in time}

In this section, we classify the growth of $\beta\left( \Psi_{N,t}, \varphi_t, \alpha_t \right)$ in time. The first two inequalities of Theorem \ref{theorem: Pauli main theorem} then follow from \eqref{eq:relation beta-a and reduced density}, \eqref{eq:relation beta-b and number operator} and the statement below. By similar estimates as in \cite[Chapter 7]{LP2020} one obtains \eqref{eq: main theorem 3} and \eqref{eq: main theorem 4}.

\begin{lemma}
\label{lemma:estimate time derivative of beta}
Let $(\varphi_t, \alpha_t)$ and  $\Psi_{N,t}$ be the solutions of \eqref{eq:Hatree-Maxwell} and \eqref{eq:Schroedinger equation microscopic} from Theorem \ref{theorem: Pauli main theorem}.
Then, there exists a monotone increasing function $C(s)$ of the norms $\norm{\varphi_s}_{H^2(\mathbb{R}^3)}$, $\norm{\abs{\cdot}^{1/2} \alpha_s}_{\mathfrak{h}}$, $\norm{v}_{L^2 + L^{\infty}(\mathbb{R}^3)}$ and 
$\norm{\left( \abs{\cdot}^{-1/2} + \abs{\cdot}^{-1} \right) \mathcal{F}[\kappa]}_{L^2(\mathbb{R}^3)}$
such that 
\begin{align}
\label{eq:estimate time derivative of beta}
\beta \left( \Psi_{N,t}, \varphi_t , \alpha_t \right) - \beta \left( \Psi_{N,0}, \varphi_0 , \alpha_0 \right)
&\leq \int_0^t ds \, C(s) \left( \beta \left( \Psi_{N,s}, \varphi_s , \alpha_s \right) + N^{-1} \right) ,
\\
\label{eq:estimate beta}
\beta \left( \Psi_{N,t}, \varphi_t , \alpha_t \right) &\leq \left(  \beta \left( \Psi_{N,0}, \varphi_0 , \alpha_0 \right) + N^{-1} \right)
e^{\int_0^t ds \, C(s)}
\end{align}
holds for any $t \geq 0$.
\end{lemma}

\begin{proof}[Sketch of the proof of Lemma \ref{lemma:estimate time derivative of beta}]
Inequality \eqref{eq:estimate time derivative of beta} is proven analogously to \cite[Lemma 6.10]{LP2020}. The similarity becomes obvious if one defines the auxiliary fields
\begin{align}
\begin{split}
\vFp(x) &\coloneqq  \frac{i}{\sqrt{2}}  \sum_{\lambda=1,2} \int d^3k \,   \vep_{\lambda}(k) 
 e^{ikx} a(k,\lambda),  
\\
\vFm(x) &\coloneqq - \frac{i}{\sqrt{2}}  \sum_{\lambda=1,2} \int d^3k \,   \vep_{\lambda}(k) 
 e^{-ikx} a^*(k,\lambda),  
\\
\vFpc(x,t) &\coloneqq  \frac{i}{\sqrt{2}}  \sum_{\lambda=1,2} \int d^3k \,   
 \vep_{\lambda}(k)  e^{ikx} \alpha_t(k,\lambda),  
\\
\vFmc(x,t) &\coloneqq - \frac{i}{\sqrt{2}}  \sum_{\lambda=1,2} \int d^3k \,    \vep_{\lambda}(k)  e^{-ikx} \overline{\alpha_t(k,\lambda)} .
\end{split}
\end{align}
By means of the cutoff function $\eta$ with Fourier transform
\begin{align}
\mathcal{F}[\eta](k) = (2 \pi)^{- 3/2} \abs{k}^{-1/2} \mathcal{F}[\kappa](k)
\end{align}
we can write the quantum and classical vector potentials as
\begin{align}
\begin{split}
\vA(x) &= - i \big( \eta * \vFp \big)(x) + i \big( \eta * \vFm \big)(x) ,
\\
\vAc(x,t) &= - i \big( \eta * \vFpc \big)(x,t) + i \big( \eta * \vFmc \big)(x,t) .
\end{split}
\end{align}
These relations are the analogue of \cite[Lemma 6.1]{LP2020}. Thus if we replace $\vE^{\sharp}$, $\vEc^{\sharp}$ in the original estimates of \cite{LP2020} by $\vF^{\sharp}$, $\vFc^{\sharp}$ with $\sharp \in \{ - , + \}$ and use
\begin{align}
\int d^3y \, \norm{\left( N^{-1/2} \vFp(y) - \vFpc(y,t) \right) \Psi_{N,t}}_{\mathcal{H}^{(N)}}^2 = 4 \pi^3 \beta^b(t)
\end{align}
we obtain \eqref{eq:estimate time derivative of beta} by similar means. This implies \eqref{eq:estimate beta} because of Gr\"onwall's inequality.

\end{proof}

\appendix

\section{Properties of \eqref{eq:Hatree-Maxwell}}
\label{section:properties of the solutions of Hatree-Maxwell}

\begin{proof}[Proof of Corollary \ref{corollary:Maxwell-Schroedinger mode function}]
For the initial data of the corollary we have that $(\varphi_0,\mathbf{A}_0,\mathbf{E}_0) \in H^2 \times H^2 \times H^1$. The existence of a unique global solution $\left( \varphi, \alpha \right)$ with $\varphi_t \in H^2(\mathbb{R}^3)$ and $\left( \abs{\cdot}^{1/2} + \abs{\cdot}^{3/2} \right) \alpha_t \in \mathfrak{h}$ then follows from Proposition \ref{prop:1}. In order to see that $\alpha_t \in \mathfrak{h}$ we bound the integral version of \eqref{eq:Hatree-Maxwell} by
\begin{align}
\label{eq:estimate alpha-t}
\norm{\alpha_t}_{\mathfrak{h}}
&\leq \norm{\alpha_0}_{\mathfrak{h}}
+ \int_0^t ds \,
\left(
\norm{\abs{\cdot} \alpha_s}_{\mathfrak{h}}
+ C  \norm{ \abs{\cdot}^{-1/2} \mathcal{F}[\kappa] \vep \cdot\mathcal{F}[\vj_s] }_{\mathfrak{h}}
\right) .
\end{align}
Using H\"older's inequality and Young's inequality we get
\begin{align}
&\norm{ \abs{\cdot}^{-1/2} \mathcal{F}[\kappa] \vep \cdot\mathcal{F}[\vj_s] }_{\mathfrak{h}}
\nonumber \\
&\quad 
\leq C \norm{\abs{\cdot}^{-1/2} \mathcal{F}[\kappa]}_{L^2(\mathbb{R}^3)} 
\left(
\norm{\varphi_s \nabla \varphi_s}_{L^{1}(\mathbb{R}^3, \mathbb{C}^3)}
+
\norm{\abs{\varphi_s}^2 \kappa * \mathbf{A}_s}_{L^{1}(\mathbb{R}^3, \mathbb{C}^3)}
\right)
\nonumber \\
&\quad 
\leq C \norm{\abs{\cdot}^{-1/2} \mathcal{F}[\kappa]}_{L^2(\mathbb{R}^3)} \norm{\varphi_s}_{H^1(\mathbb{R}^3)}^2
\left(
1 +
\norm{\left( - \Delta \right)^{-1/4} \kappa}_{L^2(\mathbb{R}^3)}
\norm{\left( - \Delta \right)^{1/4} \mathbf{A}_s}_{L^2(\mathbb{R}^3)}
\right)
\nonumber \\
&\quad \leq 
C  \norm{\varphi_s}_{H^1(\mathbb{R}^3)}^2 \left( \norm{\mathbf{A}_s}_{H^{\frac{1}{2}}(\mathbb{R}^3, \mathbb{C}^3)}   + 1 \right) \left( \norm{\abs{\cdot}^{-1/2} \mathcal{F}[\kappa]}_{L^2(\mathbb{R}^3)}^2 + 1
\right) .
\end{align}
With the help of Proposition \ref{prop:1} we conclude that the right hand side of \eqref{eq:estimate alpha-t} is finite. This shows the claim.
\end{proof}

\section*{Acknowledgments}
N.L. would like to thank Peter Pickl for many fruitful discussions within the project \cite{LP2020}. 
\\
M.F. acknowledges support from ``Istituto Nazionale di Alta Matematica
(INdAM)'' through the ``Progetto Giovani GNFM 2020: Emergent Features in Quantum Bosonic
Theories and Semiclassical Analysis''.
N.L. gratefully acknowledges support from the Swiss National Science
Foundation through the NCCR SwissMap and funding from the European Union's
Horizon 2020 research and innovation programme under the Marie Sk\l
odowska-Curie grant agreement N\textsuperscript{o} 101024712.

{}


\begin{thebibliography}{11}
  \addcontentsline{toc}{chapter}{Bibliography}

\bibitem{AF2014}
Z.~Ammari and M.~Falconi.
Wigner measures approach to the classical limit of the Nelson model:
convergence of dynamics and ground state energy.
\emph{J. Stat. Phys.} 157(2), 330--362 (2014).

\bibitem{AF2017}
Z.~Ammari and M.~Falconi.
Bohr's correspondence principle for the renormalized Nelson model.
\emph{SIAM J. Math. Anal.} 49(6), 5031--5095 (2017).

\bibitem{AFH2022}
Z. Ammari, M. Falconi and F. Hiroshima.
Towards a derivation of Classical ElectroDynamics of charges and fields from QED. \emph{Preprint},
arXiv:2202.05015  (2022).



\bibitem{bejenaru}
I. Bejenaru and D. Tataru.
Global wellposedness in the energy space for the Maxwell-Schr\"{o}dinger system.
\emph{Comm. Math. Phys.}  288(1), 145--198 (2009). 

\bibitem{CCFO2019} R.~Carlone, M.~Correggi, M.~Falconi and M.~Olivieri.
  Microscopic derivation of time-dependent point interactions. \emph{SIAM
    J. Math. Anal.} 53(4), 4657--4691 (2021).


\bibitem{CF2018} M.~Correggi and M.~Falconi. Effective potentials generated by field interaction in the quasi-classical  limit.
\emph{Ann. Henri Poincar\'e} 19(1), 189--235 (2018).
  
\bibitem{CFO2019} M.~Correggi, M.~Falconi and M.~Olivieri.  Quasi-classical
  dynamics.  \emph{J. Eur. Math. Soc.}, to appear. \emph{Preprint},
  \href{https://arxiv.org/abs/1909.13313}{arXiv:1909.13313} (2019).

\bibitem{CFO2020} M.~Correggi, M.~Falconi and M.~Olivieri.  Ground state
  properties in the quasi-classical regime. \emph{Preprint},
  \href{https://arxiv.org/abs/2007.09442}{arXiv:2007.09442} (2020).

\bibitem{D1979}
E.\,B.~Davies.
Particle-boson interactions and the weak coupling limit.
\emph{J. Math. Phys.} 20, 345--351 (1979).



\bibitem{F2013}
M.~Falconi.
Classical limit of the Nelson model with cutoff.
\emph{J. Math. Phys.} 54(1), 012303 (2013).

\bibitem{FLMP2021}
M. Falconi, N. Leopold, D. Mitrouskas and S. Petrat.
Bogoliubov Dynamics and Higher-order Corrections for the Regularized Nelson Model.
\emph{Preprint}, arXiv:2110.00458 (2021).

\bibitem{FRS2021} D.~Feliciangeli, S.~Rademacher and R.~Seiringer. Persistence of the spectral gap for the Landau--Pekar equations. \emph{Lett. Math. Phys.} 111(19) (2021).

\bibitem{FG2017}
R.\,L.~Frank and Z.~Gang.
Derivation of an effective evolution equation for a strongly coupled polaron.
\emph{Anal. PDE} 10(2), 379--422 (2017).

\bibitem{FS2014}
R.\,L.~Frank and B.~Schlein.
Dynamics of a strongly coupled polaron.
\emph{Lett. Math. Phys.} 104, 911--929 (2014).

\bibitem{GV1979}
J. Ginibre and G. Velo.
The Classical Field Limit of Scattering Theory for Non-Relativistic Many-Boson Systems. I . 
\emph{Comm. Math. Phys.} 66, 37--76 (1979).


\bibitem{GNV2006}
J. Ginibre, F.  Nironi and G. Velo.
Partially classical limit of the Nelson model.
\emph{Ann. H. Poincar\'e} 7, 21--43 (2006).

\bibitem{G2017}
M.~Griesemer.
On the dynamics of polarons in the strong-coupling limit.
\emph{Rev. Math. Phys.} 29(10), 1750030 (2017).

\bibitem{hiroshima}
F. Hiroshima. Self-adjointness of the Pauli-Fierz Hamiltonian for arbitrary values of coupling constants.
\emph{Ann. H. Poincar\'e} 3, 171--201 (2002).


\bibitem{K2009}
A. Knowles.
Limiting dynamics in large quantum systems.
\emph{Ph.D thesis},
\url{http://www.unige.ch/~knowles/thesis.pdf}  (2009).

\bibitem{KP2009} A.~Knowles and P. Pickl. Mean-field dynamics: singular potentials and rate of convergence. \emph{Commun. Math. Phys.} 298, 101--138 (2009).

\bibitem{LMRSS2021}
N.~Leopold, D.~Mitrouskas, S.~Rademacher, B.~Schlein and R.~Seiringer. Landau--Pekar equations and quantum fluctuations for the dynamics of a strongly coupled polaron. \emph{Pure Appl. Anal} 3(4), 653--676 (2021).

\bibitem{LMS2021} N.~Leopold, D.~Mitrouskas and R.~Seiringer. Derivation of the Landau--Pekar equations in a many-body mean-field limit. \emph{Arch. Ration. Mech. Anal.} 240, 383--417 (2021).

\bibitem{LP2019} N.~Leopold and S.~Petrat. Mean-field dynamics for the Nelson model with fermions. \emph{Ann. Henri Poincar{\'e}} 20(10), 3471--3508 (2019).

\bibitem{LP2018}
N.~Leopold and P.~Pickl.
Mean-field limits of particles in interaction with quantized radiation fields.
In: D.~Cadamuro, M.~Duell, W.~Dybalski, and S.~Simonella (eds) \emph{Macroscopic Limits of Quantum Systems}, volume 270 of Springer Proceedings in Mathematics \& Statistics, 185--214 (2018).

\bibitem{LP2020}
N.~Leopold and P.~Pickl.
Derivation of the Maxwell--Schr\"odinger equations from the Pauli--Fierz Hamiltonian. 
\emph{SIAM J. Math. Anal.} 52(5), 4900--4936 (2020).

\bibitem{LRSS2019}
N.~Leopold, S.~Rademacher, B.~Schlein and R.~Seiringer.
The Landau--Pekar equations: adiabatic theorem and accuracy.
\emph{ Anal. \& PDE} 14(7), 2079--2100 (2021).


\bibitem{matte} O. Matte.
Pauli-Fierz Type Operators with Singular Electromagnetic Potentials on General Domains. \emph{Math. Phys. Anal. Geom.} 20, 18 (2017). 
    
\bibitem{M2021} D.~Mitrouskas. A note on the Fr\"ohlich dynamics in the strong coupling limit. \emph{Lett. Math. Phys.} 111, 45 (2021). 



\bibitem{nakamurawada}
M. Nakamura and T. Wada.
Global Existence and Uniqueness of Solutions to the Maxwell-Schr\"odinger Equations. 
\emph{Comm. Math. Phys.} 276, 315-339 (2007).

\bibitem{P2011}
P. Pickl. A simple  derivation of mean field limits for quantum systems.
\emph{Lett. Math. Phys.} 97, 151--164 (2011).


\bibitem{RS2009}
I.~Rodnianski and B.~Schlein.
Quantum fluctuations and rate of convergence towards mean field dynamics. \emph{Commun. Math. Phys.} 291(1), 31--61 (2009).    


\bibitem{S2010}
M. \v{S}indelka. 
Derivation of coupled Maxwell-Schr\"odinger equations describing matter-laser interaction from first principles of quantum electrodynamics,
\emph{Phys. Rev. A} 81 3, 033833 (2010).


\bibitem{S2004}
H. Spohn:
Dynamics of charged particles and their radiation field,
\emph{Cambridge University Press}, Cambridge (2004),
 ISBN 0-521-83697-2.

\bibitem{T2002}
S. Teufel.
Effective N-body Dynamics for the Massless Nelson Model and Adiabatic Decoupling without Spectral Gap.
\emph{Ann. Henri Poincar\'e} 3,  939--965 (2002).




\end{thebibliography}
\end{document}